\documentclass{acm_proc_article-sp}

\usepackage{algpseudocode}
\usepackage{algorithm}
\usepackage{graphicx}
\usepackage{subfigure}
\usepackage{lineno}

\let\hat\widehat
\newtheorem{thm}{Theorem}
\newtheorem{lem}[thm]{Lemma}

\newcommand\R{\mathbb{R}}
\newcommand\E{\mathbb{E}}
\newcommand\mathand{\ {\rm and}\ }
\newcommand\norm[1]{\|#1\|}
\newcommand\Length{{\sf length}}

\catcode`@=11
\newskip\beforeproofvskip
\newskip\afterproofvskip
\beforeproofvskip=\medskipamount
\afterproofvskip=\bigskipamount

\def\prooftag{Proof}
\def\proofskip{\enspace}

\def\proof{\@ifnextchar[{\@@proof}{\@proof}}  
\def\@startproof{\par\vskip\beforeproofvskip\leavevmode}
\def\@proof{\@startproof{\scshape\prooftag.}\proofskip}
\def\@@proof[#1]{\@startproof {\scshape\prooftag #1.}\proofskip}

\catcode`@=12

\begin{document}

\title{Uncertainty Measures and Limiting Distributions for Filament Estimation}
%
%
%
%
%

\numberofauthors{3}
%
\author{
%
%
\alignauthor Yen-Chi Chen\\
       \affaddr{Carnegie Mellon University}\\
       \affaddr{5000 Forbes Ave.}\\
       \affaddr{Pittsburgh, PA 15213, USA}\\
       \email{yenchic@andrew.cmu.edu}
\alignauthor Christopher R. Genovese\\
       \affaddr{Carnegie Mellon University}\\
       \affaddr{5000 Forbes Ave.}\\
       \affaddr{Pittsburgh, PA 15213, USA}\\
       \email{genovese@stat.cmu.edu}
\alignauthor Larry Wasserman\\
       \affaddr{Carnegie Mellon University}\\
       \affaddr{5000 Forbes Ave.}\\
       \affaddr{Pittsburgh, PA 15213, USA}\\
       \email{larry@stat.cmu.edu}
}

\maketitle
\begin{abstract}
A filament is a high density, connected region in a point cloud.
There are several methods for estimating filaments
but these methods do not provide any measure of uncertainty.
We give a definition for the uncertainty of estimated filaments and
we study statistical properties of the estimated filaments.
We show how to estimate the uncertainty measures and 
we construct
confidence sets based on a bootstrapping technique. 
We apply our methods to astronomy data and 
earthquake data.
\end{abstract}



\section{Introduction}

A filament is a one-dimensional, smooth, connected structure
embedded in a multi-dimensional space.
Filaments arise in many applications.
For example, matter in the universe tends to concentrate
near filaments that comprise what is known
as the cosmic-web~\cite{Bond1996},
and the structure of that web can serve
as a tracer for estimating fundamental cosmological constants.
Other examples include neurofilaments and blood-vessel networks
in neuroscience~\cite{Lalonde2003},
fault lines in seismology~\cite{Fault},
and landmark paths in computer vision~\cite{Hile2009}.

Consider point-cloud data $X_1, X_2, \ldots, X_n$ in $\R^d$,
drawn independently from a density $p$ with compact support.
We define the filaments of the data distribution as the
\emph{ridges} of the probability density function $p$.
(See Section \ref{sec::ridges} for details.)
There are several alternative ways to formally define filaments~\cite{Eberly1996},
but the definition we use has several
useful statistical properties \cite{Genovese2012a}.
Figure~\ref{Fig:EX} shows two simple examples of point cloud data sets
and the filaments estimated by our method.

\begin{figure}
	\centering
	\subfigure
	{
		\includegraphics[scale =0.09]{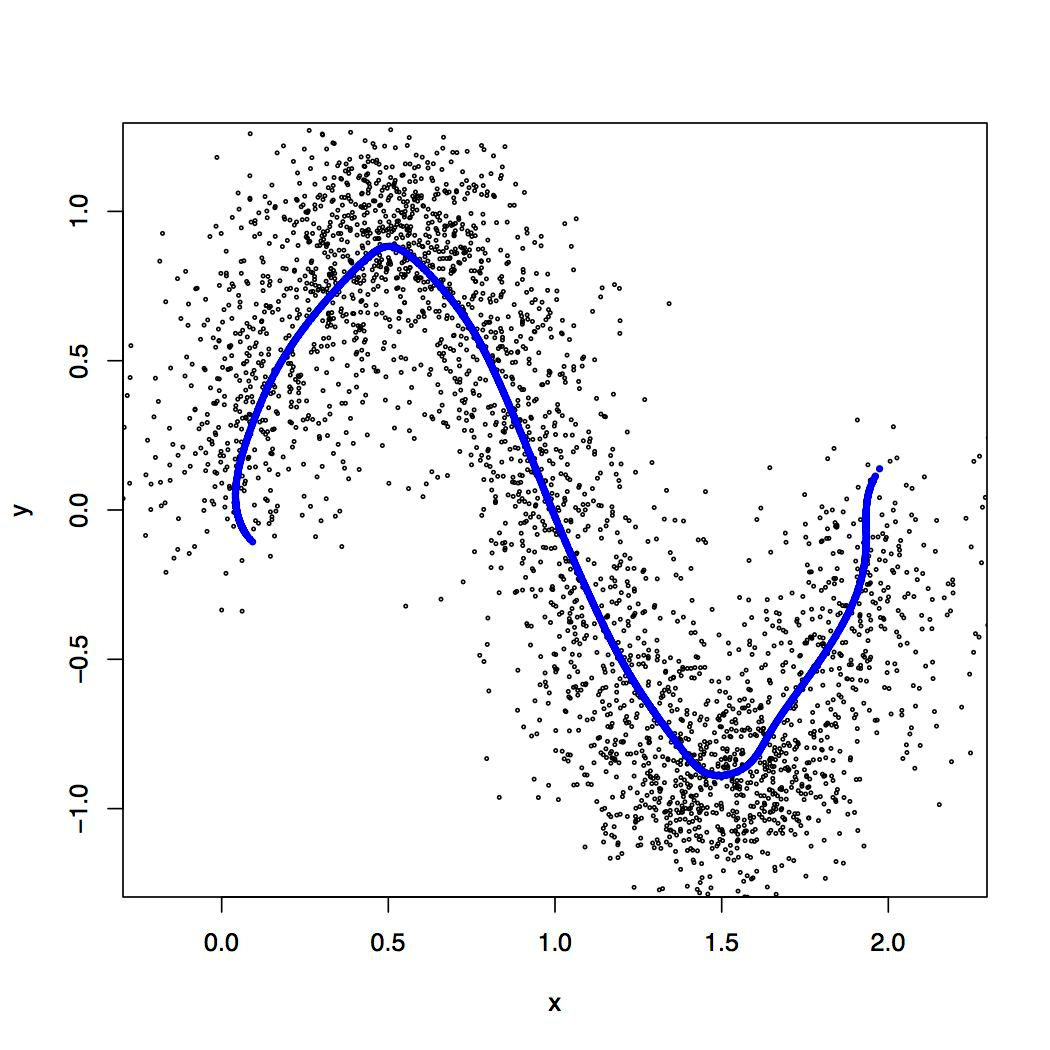}
	}
	\subfigure
	{
		\includegraphics[scale =0.09]{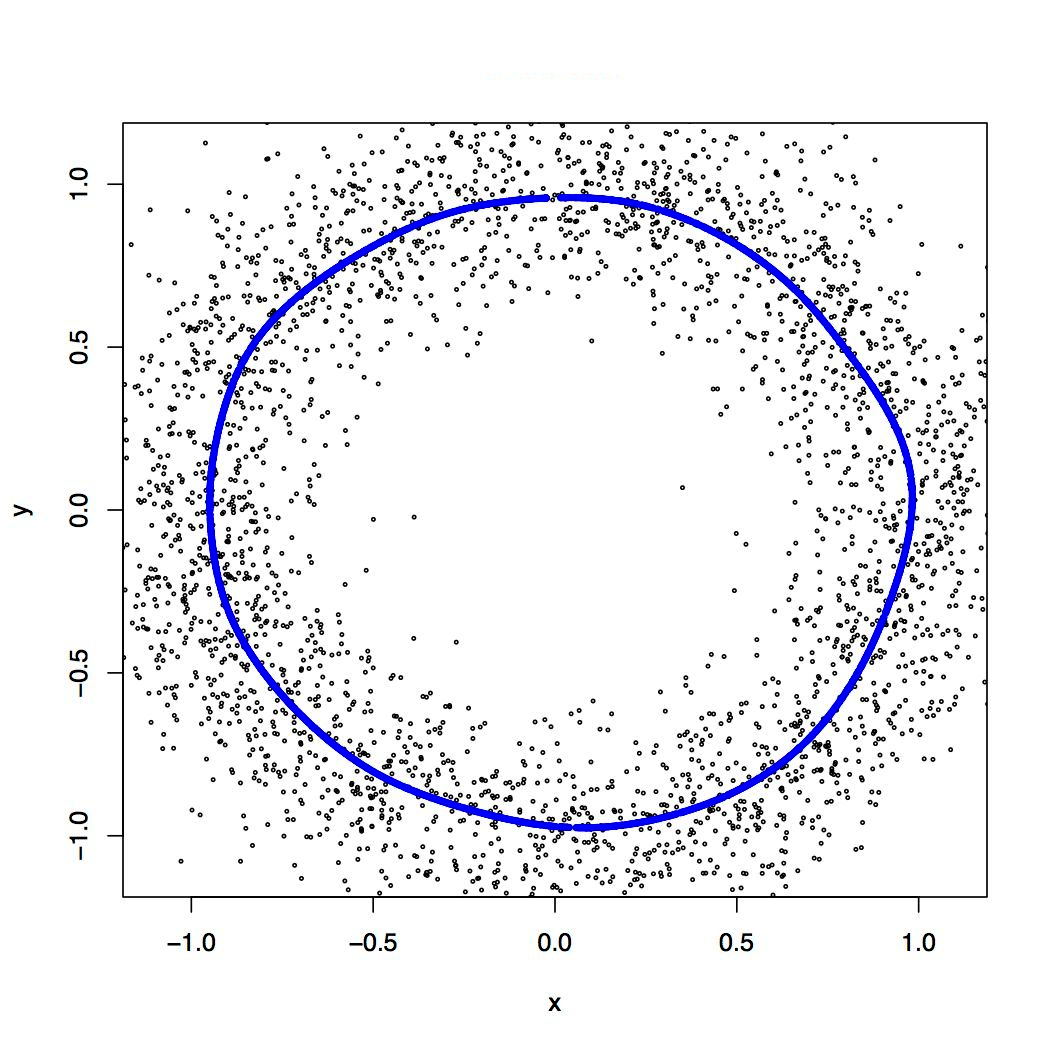}	
	}
\caption{Examples of point cloud data with ridges (filaments).}
\label{Fig:EX}
\end{figure}

The problem of estimating filaments
has been studied in several fields
and a variety of methods have been developed,
including parametric~\cite{Stoica2007,Stoica2008};
nonparametric~\cite{Genovese2012b,Genovese2010,Genovese2012c}; 
gradient based~\cite{Genovese2012a,Sousbie2011,Novikov2006}; 
and topological~\cite{Dey2006,Lee1999,Cheng2005,Aanjaneya2012,Fabrizio2013}.

While all these methods provide filament estimates,
none provide an assessment of the estimate's uncertainty.
That filament estimates are random sets
is a significant challenge
in constructing valid uncertainty measures 
\cite{Molchanov2005}. 
In this paper, we introduce a local uncertainty measure for filament
estimates.
We characterize the asymptotic distribution of estimated filaments
and use it to 
derive consistent estimates of the local uncertainty
measure and 
to construct valid confidence sets
for the filament based on bootstrap resampling.
Our main results are as follows:
\begin{itemize}
\item We show that if the data distribution is smooth, so are the estimated filaments (Theorem~\ref{S1}).
\item We find the asymptotic distribution for estimated local uncertainty and its convergence rate (Theorem~\ref{LU2},~\ref{LU4}).
\item We construct valid and consistent, bootstrap confidence sets for the local uncertainty,
      and thus pointwise confidence sets for the filament (Theorem~\ref{SB2}).
\end{itemize}
We apply our methods to point cloud data from examples in Astronomy 
and Seismology
and demonstrate that they yield useful confidence sets.

\section{Background}
\subsection{Density Ridges} \label{sec::ridges}

Let $X_1,\cdots X_n$ be random sample from a 
distribution with compact support in $\R^d$ that has density $p$.
Let $g(x) = \nabla p(x)$ and $H(x)$
denote the gradient and Hessian, respectively, of $p(x)$.
We begin by defining the \emph{ridges} of $p$,
as defined in~\cite{Genovese2012a,Ozertem2011,Eberly1996}.
While there are many possible definitions of ridges,
this definition gives stability in the underlying density, 
estimability at a good rate of convergence, and fast reconstruction algorithms,
as described in \cite{Genovese2012a}.
In the rest of this paper, the filaments to be estimated are just the one-dimensional
ridges of $p$.

A mode of the density $p$ -- where the gradient $g$ is zero and all the eigenvalues of $H$
are negative -- can be viewed as a zero-dimensional ridge.
Ridges of dimension $0 < s < d$ 
generalize this to the zeros of a \emph{projected gradient}
where the $d - s$ smallest eigenvalues of $H$ are negative.
In particular for $s = 1$,
\begin{equation} \label{eq::ridge-def}
R \equiv \mbox{Ridge}(p) = \{x: G(x) = 0, \ \lambda_2(x) < 0\},
\end{equation}
where 
\begin{align}
G(x) = V(x)V(x)^T g(x)
\end{align}
is the projected gradient.
Here, the matrix $V$ is defined as $V(x)=[v_2(x),\cdots v_d(x)]$
for eigenvectors $v_1(x),v_2(x),...,v_d(x)$ of $H(x)$ 
corresponding to eigenvalues
$\lambda_1(x) \ge \lambda_2(x)\ge \cdots \ge \lambda_d(x)$
Because one-dimensional ridges are the primary concern of this paper,
we will refer to $R$ in (\ref{eq::ridge-def}) as the ``ridges'' of $p$.

Intuitively, at points on the ridge,
the gradient is the same as the largest eigenvector
and the density curves downward sharply in directions orthogonal to that.
When $p$ is smooth and 
the \emph{eigengap}
$\beta(x) = \lambda_1(x) - \lambda_2(x)$
is positive,
the ridges have the all essential properties of filaments.
That is, $R$ decomposes into a set of smooth curve-like
structures with high density and connectivity.
$R$ can also be characterized through Morse theory~\cite{Guest2001}
as the collection of $(d-1)$-critical-points
along with the local maxima,
also known as the set of 1-ascending manifolds with their local-maxima limit points~\cite{Sousbie2011}.

\subsection{Ridge Estimation}

We estimate the ridge in three steps: density estimation, thresholding, and ascent.
First, we estimate $p$ from from the data $X_1,\ldots, X_n$.
Here, we use the well-known kernel density estimator (KDE) defined by
\begin{align}
\hat{p}_n(x) = \frac{1}{nh^d} \sum_{i=1}^n K\left(\frac{||x-X_i||}{h}\right),
\end{align}
where the kernel $K$ is a smooth, symmetric density function such as a Gaussian
and $h\equiv h_n>0$ is the bandwidth which controls the smoothness of the estimator.
Because ridge estimation can tolerate a fair degree of oversmoothing
(as shown in \cite{Genovese2012a}), we select $h$ by a simple rule
that tends to oversmooth somewhat, the multivariate Silverman's rule~\cite{Silverman1986}.
Under weak conditions, this estimator is consistent;
specifically, $||\hat p_n - p||_\infty\stackrel{P}{\to}0$ as $n\to \infty$.
(We say that $X_n$ converges in probability to $b$, written
$X_n \stackrel{P}{\to}b$ if, for every $\epsilon>0$,
$P(|X_n - b|>\epsilon)\to 0$ as $n\to\infty$.)

Second, we threshold the estimated density to eliminate low-probability regions
and the spurious ridges produced in $\hat p_n$ by random fluctuations.
Here, we remove points with estimated density less than $\tau ||\hat{p}_n||_{\infty}$
for a user-chosen threshold $0 < \tau < 1$.

Finally, for a set of points above the density threshold,
we follow the ascent lines of the projected gradient to the ridge,
which is the the subspace constrainted mean
shift (SCMS) algorithm~\cite{Ozertem2011}.
This procedure can be viewed as estimating the ridge by applying the Ridge operator to $\hat p_n$:
\begin{align}
\hat{R}_n = \mbox{Ridge}(\hat{p}_n).
\end{align}
Note that $\hat{R}_n$ is a random set.

\subsection{Bootstrapping and Smooth Bootstrapping}

The bootstrap~\cite{Efron1979} is a statistical method for
assessing the variability of an estimator.
Let $X_1,\ldots, X_n$ be a random sample from a distribution $P$
and let $\theta (P)$ be some functional of $P$ to be estimated,
such as the mean of the distribution or (in our case) the ridge set of its density.
Given some procedure $\hat\theta(X_1,\ldots, X_n)$ for estimating $\theta(P)$
we estimate the variability of $\hat\theta$ by \emph{resampling} from
the original data.

Specifically, we draw a \emph{bootstrap sample}
$X_1^*,\ldots, X_n^*$ independently and with replacement
from the set of observed data points $\{X_1,\ldots,X_n\}$
and compute the estimate $\hat\theta^* = \hat\theta(X_1^*,\ldots, X_n^*)$
using the bootstrap sample as if it were the data set.

This process is repeated $B$ times, yielding $B$ bootstrap
samples and corresponding estimates $\hat\theta^*_1,\ldots,\hat\theta^*_B$.
The variability in these estimates is then used to assess the variability
in the original estimate $\hat\theta \equiv \hat\theta(X_1,\ldots, X_n)$.
For instance, if $\theta$ is a scalar, the variance of $\hat\theta$
is estimated by 
$$
\frac{1}{B}\sum_{b=1}^B (\hat\theta_b^* - \overline{\theta})^2
$$
where
$\overline{\theta} = \frac1B \sum_{b=1}^B \hat\theta_b^*$.
Under suitable conditions, it can be shown that
this bootstrap variance estimates -- and confidence sets produced from it --
are consistent.

The \emph{smooth bootstrap} is a variant of the bootstrap that can be useful in function estimation problems
where the same procedure is used except
the bootstrap sample is drawn from the estimated density $\hat p$
instead of the original data. We use both variants below.

\section{Methods}

We measure the \emph{local uncertainty} in a filament (ridge) estimator $\hat R_n$
by the expected distance between a specified point in the original filament $R$
and the estimated filament:
\begin{equation} \label{eq::local-uncertainty}
\rho^2_n(x) = \begin{cases}
                      \E_{p} d^2(x, \hat{R}_n) & \mbox{if $x \in R$} \\
                      0                       & \mbox{\rm otherwise} 
              \end{cases},
\end{equation}
where $d(x, A)$ is the distance function:
\begin{align}
d(x,A) = \underset{y\in A}{\inf}|x-y|.
\label{M:eq1}
\end{align}
The local uncertainty measure can be understood as the expected
dispersion for a given point in the original filament to the estimated
filament based on sample with size $n$. The theoretical analysis of $\rho^2_n(x)$ is
given in theorem~\ref{LU4}.



\subsection{Estimating Local Uncertainty}

Because $\rho^2_n(x)$ is defined in terms unknown distribution $p$ and the unknown filament set $R$,
it must be estimated.
We use bootstrap resampling to do this, defining an estimate of local uncertainty
\emph{on the estimated filaments}.
For each of $B$ bootstrap samples, $X_1^{*(b)},\cdots,X^{*(b)}_n$,
we compute the kernel density estimator
$\hat{p}^{*(b)}_n$,
the ridge estimate
$\hat{R}^{*(b)}_n = \mbox{Ridge}(\hat{p}^{*(b)}_n)$,
and the divergence
$\rho^{2}_{(b)}(x) = d^2(x, \hat{R}^{*(b)}_n)$ for all $x\in\hat R_n$.
We estimate $\rho^2_n(x)$ by
\begin{align}
\hat{\rho}^2_n(x) = \frac1B \sum_{b=1}^B \rho^{2}_{(b)}(x) \,\equiv \E(d^2(x, \hat{R}^*_n)|X_1,\cdots,X_n),
\end{align}
for each $x \in \hat{R}_n$,
where the expectation is from the (known) bootstrap distribution.
Algorithm 1 provides pseudo-code for this procedure,
and Theorem~\ref{SB2} shows that the estimate is consistent
under smooth bootstrapping.

\begin{algorithm}
\caption{Local Uncertainty Estimator}
\begin{algorithmic}
\State \textbf{Input:} Data $\{ X_1,\ldots,X_n\}$.
\State 1. Estimate the filament from $\{ X_1,\ldots,X_n\}$; denote the estimate by $\hat{R}_n$.
\State 2. Generate $B$ bootstrap samples: $X^{*(b)}_1,\ldots,X^{*(b)}_n$ for $b = 1,\ldots,B$.
\State 3. For each bootstrap sample, estimate the filament, yielding $\hat{R}^{*(b)}_n$ for $b = 1,\ldots,B$.
\State 4. For each $x\in \hat{R}_n$, calculate $\rho^{2}_{(b)}(x) = d^2(x,\hat{R}^{*(b)}_n)$, $b=1,\ldots,B$.
\State 5. Define $\hat{\rho}^2_n(x) = \mbox{mean}\{r^2_1(x),\ldots,r^2_B(x)\}$.
\smallskip
\State \textbf{Output:} $\hat{\rho}^2_n(x)$.
\end{algorithmic}
\end{algorithm}

\subsection{Pointwise Confidence Sets}

Confidence sets provide another useful assessment of uncertainty.
A $1-\alpha$ confidence set is a random set computed from the data
that contains an unknown quantity with at least probability $1 - \alpha$.
We can construct a pointwise confidence set for filaments from the
distance function ~\eqref{M:eq1}. For each point $x\in\hat{R}_n$, let
$r_{1-\alpha}(x)$ be the
$(1-\alpha)$ quantile value of $d(x, \hat{R}^*)$ 
from the bootstrap.
Then, define
\begin{align}
C_{1-\alpha}(X_1,\cdots,X_n) = \bigcup_{x\in \hat{R}} B(x,r_{1-\alpha}(x)).
\end{align} 
This confidence set capture the local uncertainty:
for a point $x\in\hat{R}_n$ with low (high) local
uncertainty, the associated radius $r_{1-\alpha}(x)$ is
small (large). 
But note that the confidence set attains $1-\alpha$ coverage around each point;
the coverage of the entire filament set is lower. That is, we can have high probability to cover each point but the probability to simultaneously cover all points (the whole filament set) might be lower.

\begin{algorithm}
\caption{Pointwise Confidence Set}
\begin{algorithmic}
\State \textbf{Input:} Data $\{ X_1,\ldots,X_n\}$; significance level $\alpha$.
\State 1. Estimate the filament from $\{ X_1,\ldots,X_n\}$; denote this by $\hat{R}_n$.
\State 2. Generates bootstrap samples $\{ X^{*(b)}_1,\ldots,X^{*(b)}_n\}$ for $b = 1,\ldots,B$.
\State 3. For each bootstrap sample, estimate the filament, call this $\hat{R}^{*(b)}_n$.
\State 4. For each $x\in \hat{R}_n$, calculate $\rho^{2}_{(b)}(x) = d^2(x,\hat{R}^{*(b)}_n)$, $b=1,\ldots,B$.
\State 5. Let $r_{1-\alpha}(x)= Q_{1-\alpha}(\rho_{(1)}(x),\ldots,\rho_{(B)}(x))$.
\smallskip
\State \textbf{Output:} $\bigcup_{x \in \hat{R}_n} B(x, r_{1-\alpha}(x))$ where $B(x, r)$ is the closed ball with center $x$ and radius $r$.
\end{algorithmic}
\end{algorithm}

\section{Theoretical analysis}

For the filament set $R$, we assume that it can be decomposed into a finite partition 
\[
\{R_1,\cdots,R_k\}
\] 
such that each $R_i$ is a one dimensional manifold. Such a partition can be constructed by the equation of traversal in page 56 of~\cite{Eberly1996}. For each $R_i$, we can parametrize it by a function $\phi_i(s): [0,1]\to R_i$ from the equation of traversal mentioned with suitable scaling. 

For simplicity, in the following proofs we assume that the filament set $R$ is a single $R_i$ so that we can construct the parametrization $\phi$ easily. All theorems and lemmas we prove can be applied to the whole filament set $R=\bigcup_i R_i$ by repeating the process for each individual $R_i$.

\subsection{Smoothness of Density Ridges}

To study the properties of the uncertainty estimator,
we first need to establish some results about the smoothness of the filament.
The following theorem provides conditions for smoothness of the filaments.
Let $\mathbf{C}^{k}$ denote the collection of $k$ times continuously
differentiable functions. 

\begin{thm}[Smoothness of Filaments] \label{S1}
Let $\phi(s): [0,1]\to R$ be a parameterization of filament set $R$,
and for $s_0\in[0,1]$, let $U\subset R$ be an open set containing $\phi(s_0)$.
If $p$ is $\mathbf{C}^{k}$ and the eigengap $\beta(x) > 0$ for $x\in U$,
then $\phi(s)$ is $\mathbf{C}^{k-2}$ for $s\in \phi^{-1}(U)$.
\end{thm}

Theorem~\ref{S1} says that filaments from a smooth density will be smooth.
Moreover, estimated filaments from the KDE will be smooth if the kernel function is smooth.
In particular, if we use Gaussian kernel, which is $\mathbf{C}^{\infty}$, then the
corresponding filaments will be $\mathbf{C}^{\infty}$ as well.

\subsection{Frenet Frame}

\begin{figure}
\center
\includegraphics[scale =0.5]{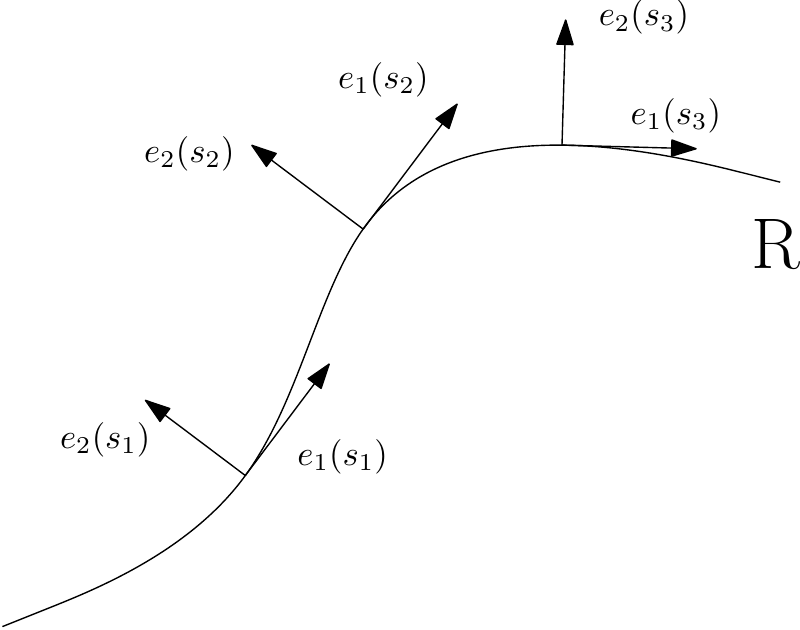}
\caption{An example for Frenet frame in two dimension.}
\label{FF1}
\end{figure}

In the arguments that follow, it is useful to have a well-defined
``moving'' coordinate system along a smooth curve.
Let $\gamma: \R \mapsto \R^d$ be an arc-length parametrization for a $\mathbf{C}^{k+1}$
curve with $k\ge d$.
The \textit{Frenet frame}~\cite{Kuhnel2002}
along $\gamma$ 
is a smooth family of orthogonal bases at $\gamma(s)$
\begin{align*}
e_1(s), e_2(s), \cdots e_d(s)
\end{align*}
such that $e_1(s) = \gamma'(s)$ determines the direction of the curve.
The other basis elements $e_2(s),\cdots, e_d(s)$ are called the curvature vectors
and can be determined by a Gram-Schmidt construction.

Assume the density is $\mathbf{C}^{d+3}$. We can construct a Frenet
frame for each point on the filaments. Let $e_1(s), \cdots, e_d(s)$ be
the Frenet frame of $\phi(s)$ such that
\begin{align*}
e_1(s) &= \frac{\phi'(s)}{|\phi'(s)|}\\
e_j(s) &= \frac{\tilde{e}_j(s)}{|\tilde{e}_j(s)|}\\
\tilde{e}_j(s) &= \phi^{(j)}(s) - \sum_{i=1}^{j-1} <\phi^{(j)}(s), e_i(s)>e_i(s), j=2,\cdots, d,
\end{align*}
where $\phi^{(j)}(s)$ is the $j$th derivative of the $\phi(s)$ and $<a,b>$ is the inner product of vector $a,b$.
An important fact is that the basis element $e_j(s)$ is
$\mathbf{C}^{d+3-j}$, $j=1,2\cdots d$.  Frenet frames are widely used in
dynamical systems because they provide a unique and continuous frame to
describe trajectories.

\subsection{Normal space and distance measure}  

The \emph{reach} of $R$, denoted by $\kappa(R)$, is the smallest real number $r$
such that each
$x\in \{y:\ d(y,R) \leq r\}$ has a unique projection onto $R$~\cite{Federer1959}.

We define the normal space $L(s)$ of $\phi(s)$ by 
\begin{align}
L(s)= \Bigl\{\sum_{i=2}^d \alpha_i e_i(s)\in \R^d:  \alpha_2^2+\cdots +\alpha_d^2 \leq \kappa(R)^2\Bigr\}.
\end{align}
Note that since we have second derivative of $\phi(s)$ exists and finite, the reach will be bounded from below.
\begin{figure}
\center
\includegraphics[scale =0.4]{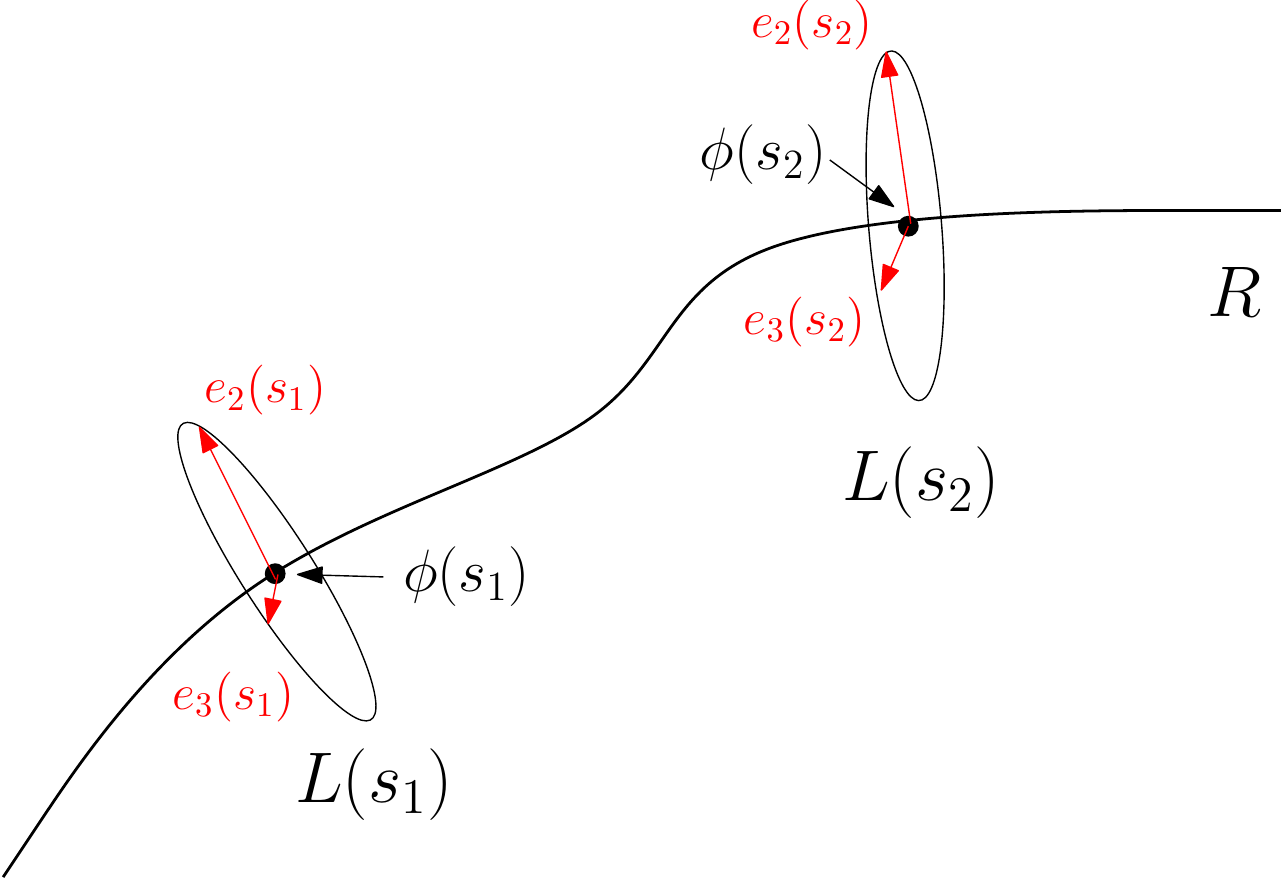}
\label{FF2}
\caption{An example for the normal space $L(s)$ along a ridge in three dimensional.}
\end{figure}

Finally, define the Hausdorff distance between two subsets of $\R^d$
by 
\begin{equation}
d_H(A,B) = \inf\{\epsilon:\; A \subset B \oplus \epsilon \mathand B \subset A \oplus \epsilon\},
\end{equation}
where $A \oplus \epsilon = \bigcup_{x\in A} B(x,\epsilon)$ and $B(x,\epsilon)=\{y:\; \norm{x-y} \le \epsilon\}$.

\subsection{Local uncertainty}

Let the estimated filament be the ridge of KDE.
We assume the following:
\begin{itemize}
\item[(K1)] The kernel $K$ is $\mathbf{C}^{d+3}$.
\item[(K2)] The kernel $K$ satisfies condition $K_1$ in page 5 of~\cite{Gine2002}.
\item[(P1)] The true density $p$ is in $\mathbf{C}^{d+3}$.
\item[(P2)] The ridges of $p$ have positive reach.
\item[(P3)] The ridges of $p$ are closed. For example, Figure~\ref{Fig:EX}-(b).
\end{itemize}
(K1) and (K2) are very mild assumptions on the kernel function. For
instance, Gaussain kernels satisfy both. (P1-P3) are
assumptions on the true density. (P1) is a smoothness condition. 
(P2) is a smoothness assumption on the ridge.
(P3) is included to
avoid boundary bias when estimating the filament near endpoints.


Now we introduce some norms and semi-norms
characterizing the smoothness of the density $p$. 
A vector $\alpha = (\alpha_1,\ldots,\alpha_d)$
of non-negative integers is called a multi-index
with $|\alpha| = \alpha_1 + \alpha_2 + \cdots + \alpha_d$
and corresponding derivative operator
$$
D^\alpha = \frac{\partial^{\alpha_1}}{\partial x_1^{\alpha_1}} \cdots \frac{\partial^{\alpha_d}}{\partial x_d^{\alpha_d}},
$$
where $D^\alpha f$ is often written as $f^{(\alpha)}$.
For $j = 0,\ldots, 4$, define
\begin{align}
  \norm{p}_{\infty}^{(j)} = \underset{\alpha:\; |\alpha| = j}{\max} \underset{x\in\R^d}{\sup} |p^{(\alpha)}(x)|.
\end{align}
When $j = 0$, we have the infinity norm of $p$; for $j > 0$, these are semi-norms.
We also define
\begin{align}
\norm{p}^*_{\infty, k} = \underset{j=0,\cdots,k}{\max} \norm{p}^{(j)}_{\infty}.
\end{align}
It is easy to verify that this is a norm.  Next we recall a theorem
in~\cite{Genovese2012a} which establish the link of Hausdorff distance
between $R,\hat{R}_n$ with the metric between density.

\begin{thm}[Theorem 6 in~\cite{Genovese2012a}] \label{LU0}
Under conditions in~\cite{Genovese2012a}, as $||p-\hat{p}_n||^*_{\infty, 3}$ is sufficiently small , we have
\begin{align*}
d_H(R,\hat{R}_n) = O_P(||p-\hat{p}_n||^*_{\infty, 2}).
\end{align*}
\end{thm}
This theorem tells us that we have convergence in Hausdorff distance for estimated filaments.

\begin{lem}[Local parametrization] \label{LU1}
For the estimated filament $\hat{R}_n$, 
define $\hat{\phi}_n(s) = L(s)\cap \hat{R}_n$
and $\Delta_n = d_H(\hat{R}_n, R)$.
Assume (K1), (K2), (P1), (P2). 
If $||p-\hat{p}_n||^*_{\infty, 4}\stackrel{P}{\to} 0$, 
then, when $\Delta_n$ is sufficiently small, \vspace{-3ex}
\begin{enumerate}
\item $\hat{\phi}_n(s)$ is a singleton for all $s$ except in a set $S_n$ containing the boundaries
  with $\Length(S_n)\leq O(\Delta_n)$.
\item $\frac{d(x,\hat{R}_n) - |\phi(s)-\hat{\phi}_n(s)|}{|\phi(s)-\hat{\phi}_n(s)|} = o_P(1)$ 
for $x$ not at the boundary of filaments.
\item If in addition (P3) holds, then $S_n = \emptyset$.
\end{enumerate}
\vspace*{-4ex}
\end{lem}
Notice that a sufficient condition for Lemma~\ref{LU1} is $\frac{nh^{d+8}}{\log n}\rightarrow \infty$ by 
Lemma~\ref{LUlem2}.

\begin{figure}
\center
\includegraphics[scale =0.4]{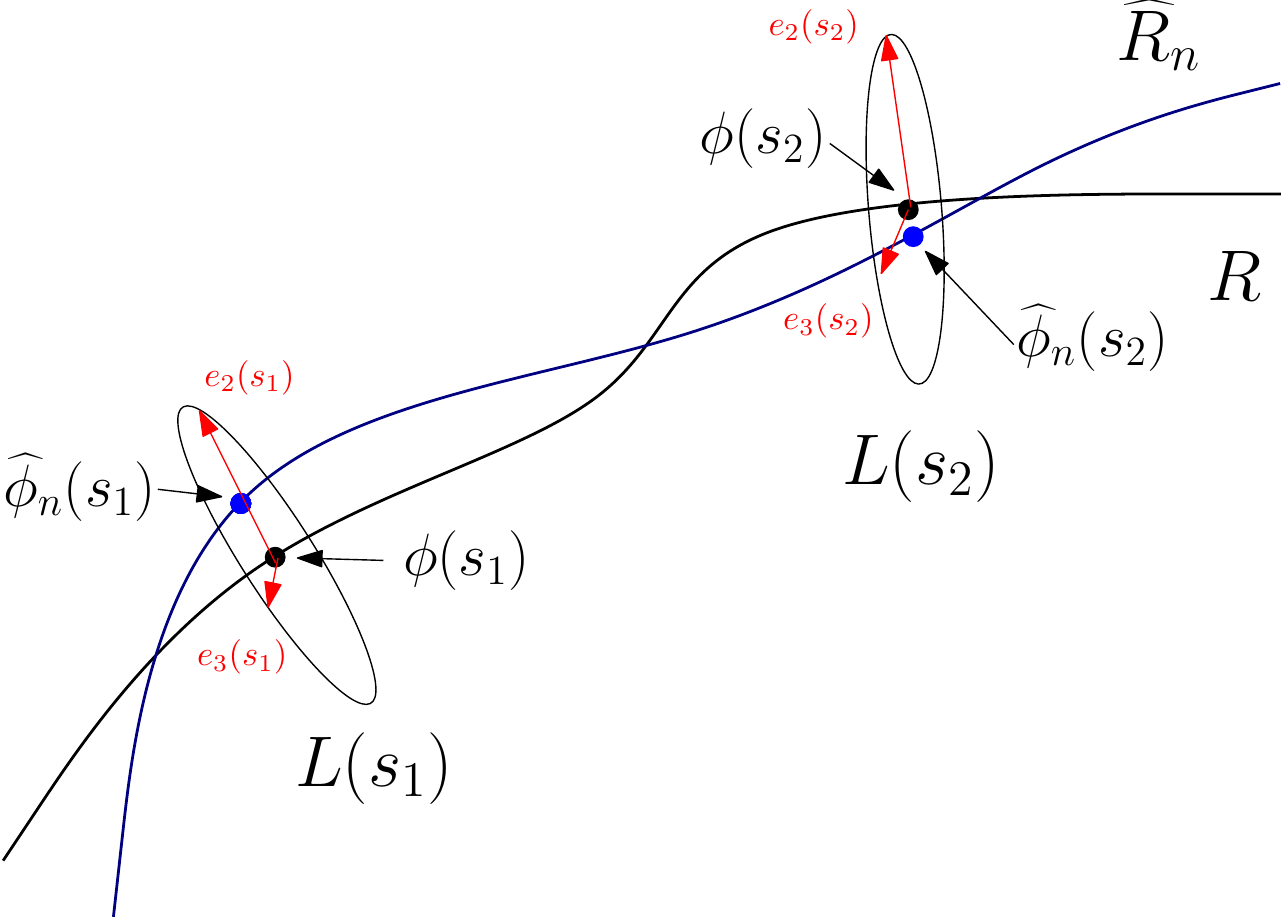}
\label{FF2}
\caption{An example for $\hat{\phi}_n(s)$.}
\end{figure}

Claim 1 follows because
the Hausdorff distance is less than $\min\{\frac{\kappa(R)}{2},
\frac{\kappa(\hat{R}_n)}{2}\}$. This will be true
since by Theorem~\ref{LU0}, the Hausdorff distance is contolled by
$||p-\hat{p}_n||^*_{\infty, 2}$, and we have a stronger convergence
assumption. The only exception is points near the boundaries of $R$
since $\hat{R}_n$ can be shorter than $R$ in this case. But this can only
occur in the set with length less than Hausdorff distance.
Claim 2 follows from the fact that the normal space for $\phi(s)$ and $\hat{\phi}$
will be asymptotically the same.
If we assume (P3), then $R$ has no boundary, so that $S_n$ is an empty set.

Note that Claim 2 gives us the validity of approximation for $d(x,\hat{R}_n)$
via $|\phi(s)-\hat{\phi}_n(s)|$. So the limiting bahavior of local uncertainty $d(x,\hat{R}_n)$
will be the same as $|\phi(s)-\hat{\phi}_n(s)|$. In the following, we will study the limiting distributions for $|\phi(s)-\hat{\phi}_n(s)|$.

We define the \emph{subspace derivative} by $\nabla_{L} = L^T\nabla$,
which in turn gives the \emph{subspace gradient}
\begin{align*}
g(x;L) = \nabla_{L} p(x)
\end{align*} 
and the \emph{subspace Hessian}
\begin{align*}
H(x;L) = \nabla_L \nabla_L p(x).
\end{align*} 
Then we have the following theorem on local uncertainty,
where $X_n\stackrel{d}{\to} Y$ denotes convergence in distribution.

\begin{thm}[Local uncertainty theorem] \label{LU2}
Assume (K1),(K2),(P1),(P2). If $\frac{nh^{d+8}}{\log n}\rightarrow \infty, nh^{d+10}\rightarrow 0$, then
\begin{align*}
\sqrt{nh^{d+2}}&([\phi(s)-\hat{\phi}_n(s)]-L(s)\mu(s) h^2) \overset{d}{\rightarrow} L(s) A(s)
\end{align*}
where
\begin{align*}
A(s) &\overset{d}{=} N(0, \Sigma(s))\in\R^{d-1}\\
\mu(s) & = c(K) H(\phi(s);L(s))^{-1} \nabla_{L(s)} (\nabla \bullet \nabla) p(\phi(s))\\
\Sigma(s)&= H(\phi(s);L(s))^{-1} \nabla_{L(s)} K (\nabla_{L(s)} K)^T H(\phi(s);L(s))^{-1} p(\phi(s))
\end{align*}
for all $\phi(s)\in R\backslash S_n$ with $\Length(S_n)\leq O(d_H(R,\hat{R}_n))$.
\end{thm}

Theorem \ref{LU2} states the asymptotic behavior of $\phi(s)-\hat{\phi}_n(s)$ which is 
asymptotically equivalent to local uncertainty. $L(s)\mu(s)h^2$ is the bias component 
and $L(s)A(s)$ is the stochastic variation component in which the parameter 
$\Sigma(s)$ controls the amount of varitaion. 
The contents in parameters $\mu(s)$ and $\Sigma(s)$ link
the geometry of the local density function with the local uncertainty.

\textbf{Remarks:} \vspace{-3ex}
\begin{itemize}
\item Note that $\frac{nh^{d+8}}{\log n}\rightarrow \infty$ is
  a sufficient condition for up to the fourth derivative uniform
  convergence. The uniform convegence in these derivative along with
  (P2) and theorem~\ref{S1} ensures the reach of
  $\hat{R}_n$ will converge to the condition number of $R$.

\end{itemize}

By theorem \ref{LU2} and claim 2 in lemma \ref{LU1}, 
we know the asymptotic distribution of local
uncertainty $d(x,\hat{R}_n) $. So we have
the following theorem on local uncertainty measure.

\begin{thm}
\label{LU4}
Define the local uncertatiny measure by
\begin{align*}
\rho^2_n(\phi(s)) = \E(d(\phi(s), \hat{R}_n)^2),
\end{align*}
where $\phi(s)$ ranges over all points in $R$.
Assume that (K1), (K2), (P1), and (P2) hold.
If $\frac{nh^{d+8}}{\log n}\rightarrow \infty, nh^{d+10}\rightarrow 0$ then
\begin{align*}
\rho^2_n(\phi(s)) = \mu(s)^T\mu(s) h^4 + &\frac{1}{nh^{d+2}} {\rm Trace}(\Sigma(s)^2) + \\
                                         &o(h^4) + o(\frac{1}{nh^{d+2}}),
\end{align*}
for all $\phi(s)\in R\backslash S_n$ with $\Length(S_n)\leq O(d_H(R,\hat{R}_n))$.
\end{thm}

This theorem is just an application of theorem~\ref{LU2}. However, it
gives the convergence rate of local uncertainty measures. 
If we assume (P3), then Theorem~\ref{LU2},~\ref{LU4} can be applied to all points on the filaments.


\subsection{Bootstrapping Result}

For the bootstrapping result, we assume (P3) for convenience. Note that
if we do not assume (P3), the result still holds for points not close
to terminals.  Let $q_m$ be a sequence of densities satisfying (P1). 
We want to study the local uncertainty of the associated filaments. So we
work on the random sample generated from $q_m$ and use the random
sample to build estimated filaments for filaments of $q_m$.  Define
$\psi_m(s), L_m^*(s)$ as the a parametrization for the filaments and
associated normal space of $q_m$. Then we have the following
convergence theorem for a sequence of densities converging to $p$.

\begin{figure}
\center
\includegraphics[scale =0.4]{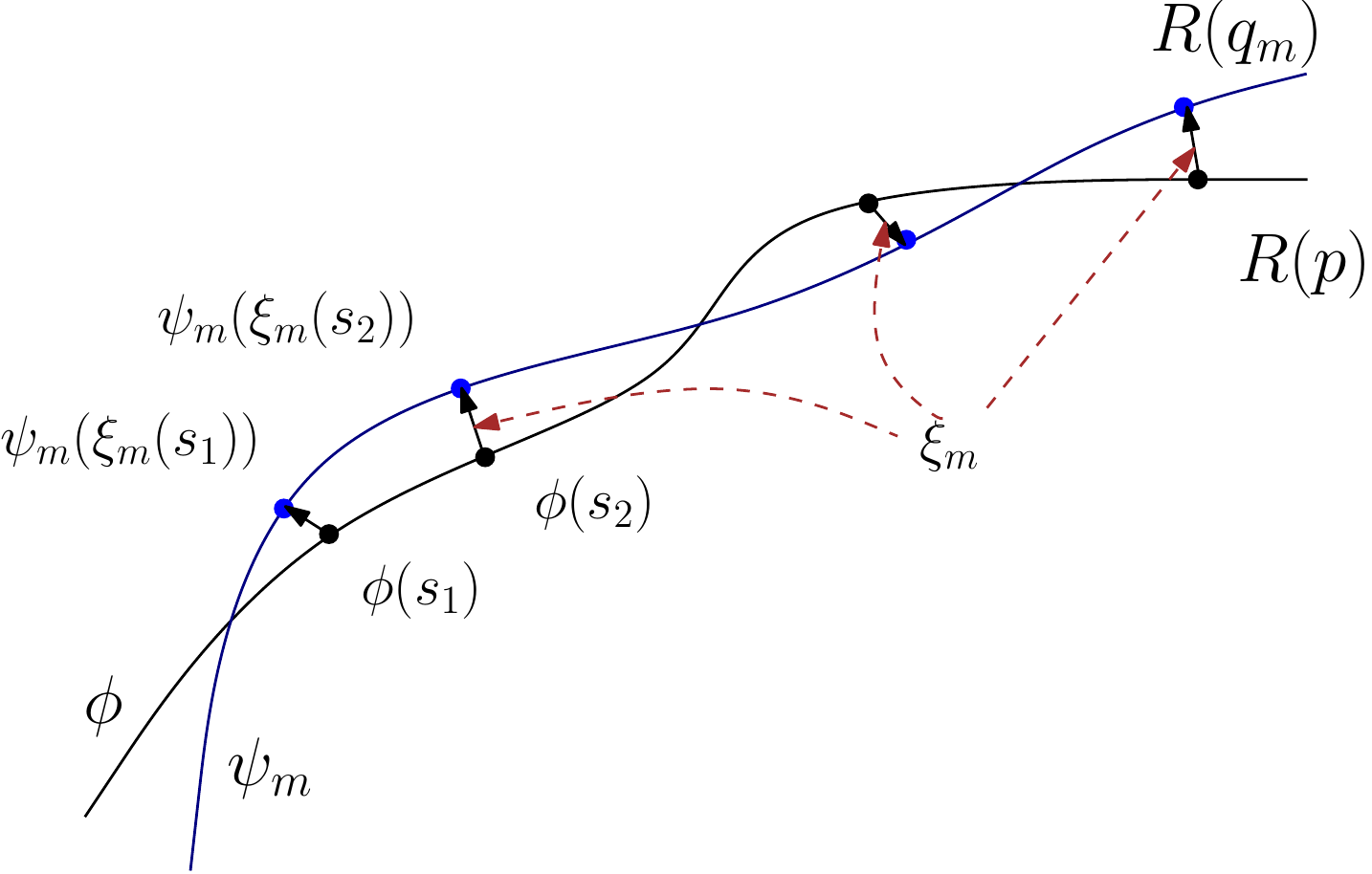}
\label{SBplot1}
\caption{An example for $\xi_m(s)$ along with $\phi, \psi_m$.}
\end{figure}

\begin{thm} \label{SB2}
  Assume that (P1--3) hold.
  Let $q_m$ be a sequence of probability densities
  that satisfy (P1), (P2), and $\norm{p - q_m}^*_{\infty,3} \to 0$ as $m\to\infty$.

  If $d_H(R(q_m), R(p))$ is sufficiently small,
  we can find a bijection $\xi_m: [0,1]\mapsto [0,1]$ such that \vspace{-3ex}
  \begin{enumerate}
  \item $|\psi_m(\xi_m(s))- \phi(s)| \to 0$.
  \item $\left|\frac{<\phi'(s),\psi'_m(\xi_m(s))>}{|\psi'_m(\xi_m(s))| |\phi'(s)|} \right| \to 1$.
  \item $\underset{s\in [0,1]}{\sup} |\mu(s;q_m)-\mu(s;p)|\to 0$.
  \item $\underset{s\in [0,1]}{\sup} |\Sigma(s;q_m)-\Sigma(s;p)|\to 0$.
  \end{enumerate}
  In particular, if we use $\hat{p}_n=q_n$ with $\frac{nh^{d+8}}{\log n}\rightarrow \infty, nh^{d+10}\rightarrow 0$,
  then the above result holds with high probability.
\end{thm}

Note that the local uncertainty measure has unknown support and unknown parameters given in theorem \ref{LU4}. Claim 1 shows the convergence in support while claim 3,4 prove the consistency of the parameters controlling uncertainties. This theorem states that if we have a sequence of densities converging to a limiting density, then the local uncertainty will converge in a sense.

\textbf{Remarks:} \vspace{-3ex}
\begin{itemize}



\item Notice that $\psi_m(\xi_m(s))$ need not be the same as
  $L(s)\cap R(q_m)$. The latter one lives in the normal space of
  $\phi(s)$ but the former need only be a continuously
  bijective mapping. The projection that maps $s$ to the point
  $L(s)\cap R(q_m)$ is one choice of $\xi_m$.

\item The last result holds immediately from Lemma~\ref{LUlem2} as we
  pick $\frac{nh^{d+8}}{\log n}\rightarrow \infty,
  nh^{d+10}\rightarrow 0$. The bandwidth in this case will ensure
  uniform convergence in probability up to the forth derivative which
  is sufficient to the condition.

\end{itemize}

\section{Examples}

We apply our methods to two datasets, one from astronomy 
and one from seismology.
In both cases, we use an isotropic Gaussian kernel for the KDE
and threshold using $\tau = 0.1$.
We use a $50\times50$ uniform grid over each sample as initial points 
in the ascent step for running SCMS.
We compare the result from bootstrapping and smooth
bootstrapping based on 100 
bootstrap samples to estimate uncertainty. 

\begin{figure}
	\centering
	\subfigure[Bootstrapping]
	{
		\includegraphics[width =2.2 in, height = 2.2in]{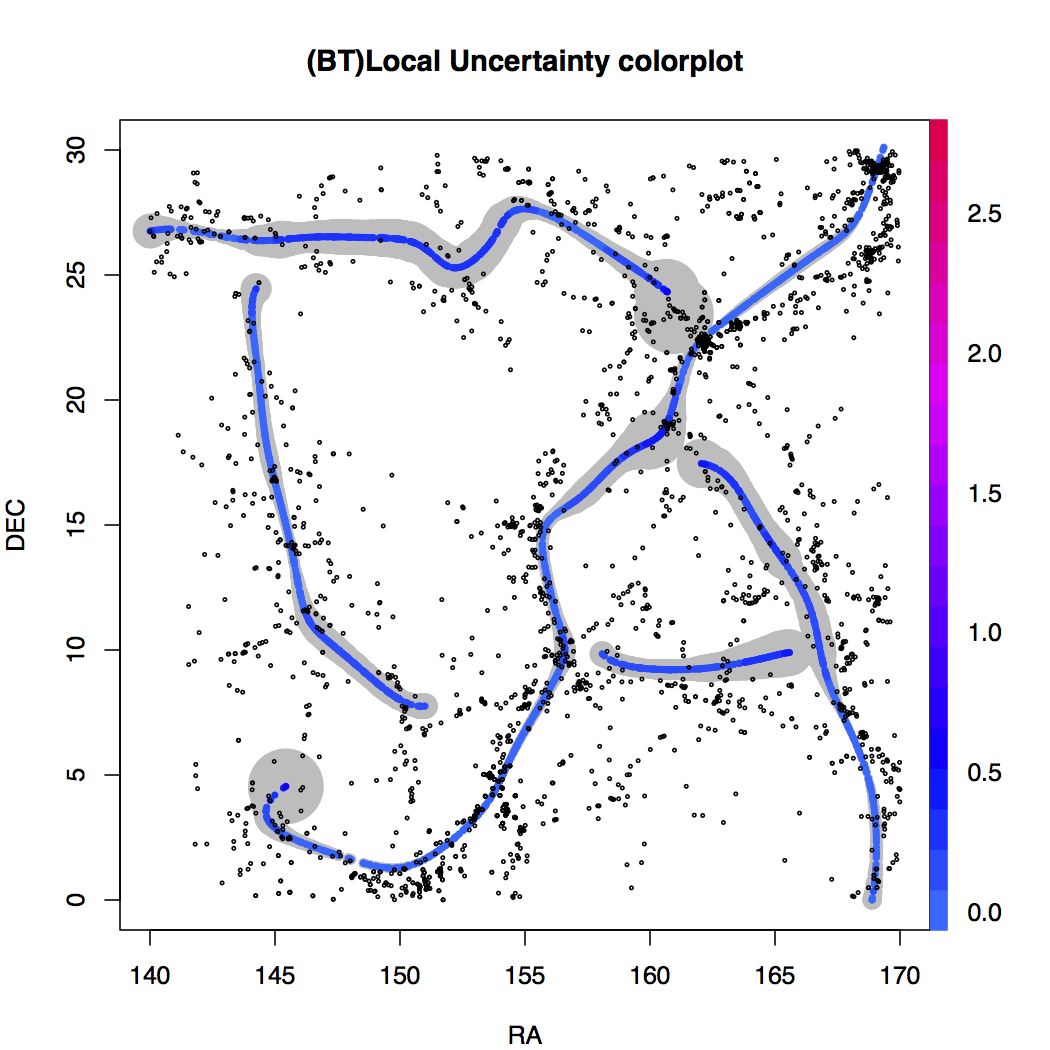}
	}
	\subfigure[Smooth bootstrapping]
	{
		\includegraphics[width =2.2 in, height = 2.2in]{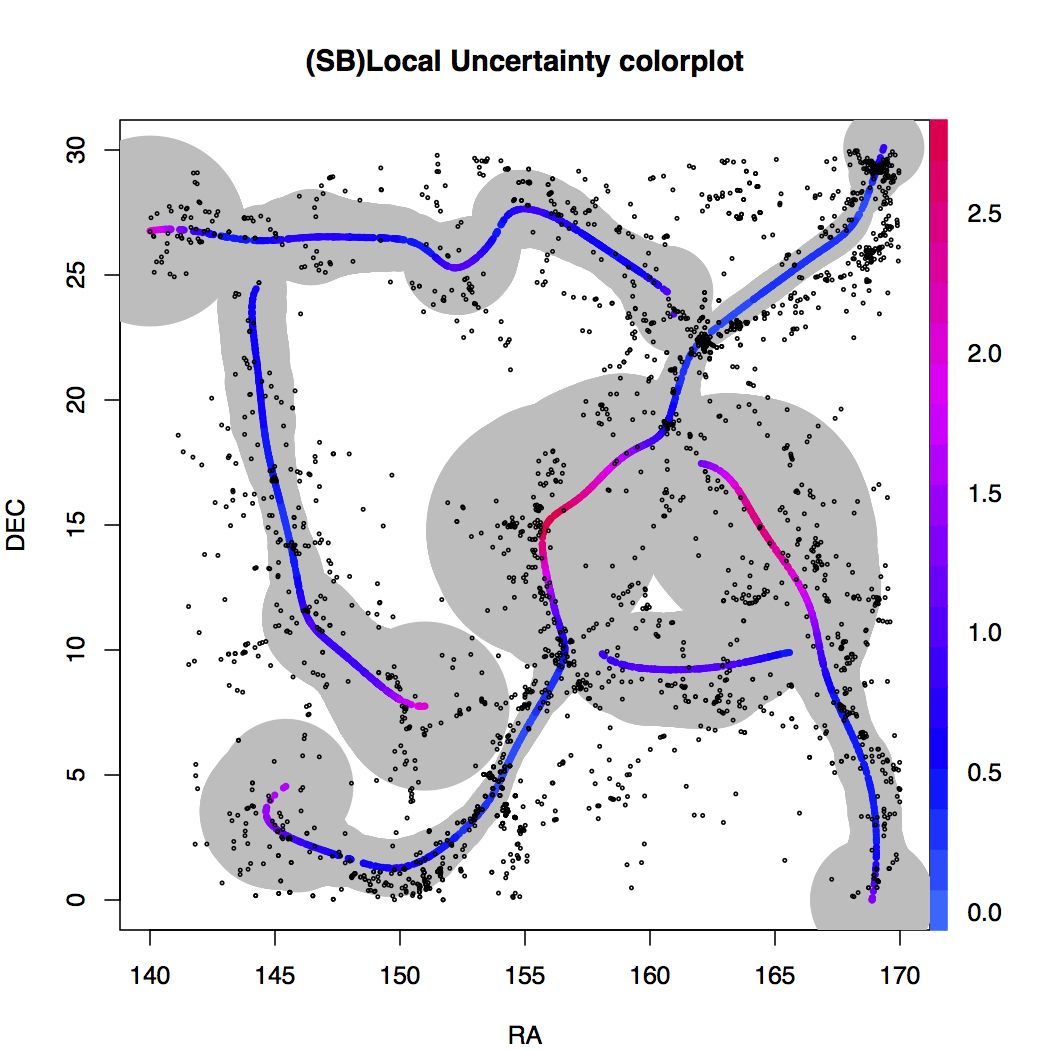}
	}
\caption{Local uncertainty measures and pointwise confidence sets for SDSS data. (a): Bootstrapping result. (b): Smooth bootstrapping result. We display local uncertainty measures based on color (red: high uncertainty) and 90\% pointwise confidence sets.}
\label{Fig:R1}
\end{figure}

\textit{Astronomy Data.}
The data come from \textit{Sloan Digit Sky Survey(SDSS) Data Release(DR) 9}.
\footnote{The SDSS dataset http://www.sdss3.org/dr9/}
In this dataset, each point is a galaxy and is characterized by three features (\textit{z, ra, dec}).
\textit{z} is the redshift value, a measurement of the distance form that galaxy to us.
\textit{ra} is right ascesion, the latitude of the sky.
\textit{dec} is declination, the longitude of the sky.

We restrict ourselves to \textit{z=0.045$\sim$0.050} which is a slice of data on the
\textit{z} coordinate that consists of $2,532$ galaxies.
We selected values in
\textit{(ra, dec)=(0 $\sim$ 30, 140 $\sim$ 170)}.
The bandwidth $h$ is 2.41.

Figure~\ref{Fig:R1} displays the local uncertainty measures with 
pointwise confidence sets. The red color indicates higher local
uncertainty while the blue color stands for lower
uncertianty. Bootstrapping shows a very small local uncertainty and 
very narrow pointwise confidence sets. Smooth bootstapping yields a
loose confidence sets but it shows a clear pattern of local
uncertainty which can be explained by our theorems.

From Figure~\ref{Fig:R1}, we identify four cases associated with high
local uncertainty: high curvature of the filament, flat density near
filaments, terminals (boundaries) of filaments, and intersecting of filaments. For
the points near curved filaments, we can see uncertainty increases in
every case. This can be explained by theorem~\ref{LU2}. The curvature
is related to the third derivative of density from the definition of
ridges. From theorem \ref{LU2}, we know the bias in filament estimation is 
proportional to the third derivative. So the estimation for highly curved filaments tends
to have a systematic bias in filament estimation and our uncertainty measure captures this
bias successfully.


For the case of a flat density, by theorem \ref{LU2}, we know both the bias
and variance of local uncertainty is proportional to the inverse of
the Hessian. A flat density has a very small Hessian matrix and thus the
inverse will be huge; this raises the uncertainty. Though our theorem
can not be applied to terminals of filaments, we can still explain the
high uncertianty. Points near terminals suffer from boundary bias in
density estimation. This leads to an increase in the uncertainty.
For regions near connections, the eigengap
$\beta(x)=\lambda_1(x)-\lambda_2(x)$ will approach $0$ which causes
instability of the ridge since our definition of ridge requires
$\beta(x)>0$. 
All cases with high local uncertainty can be explained by our theoretical result. 
So the data analysis is consistent with our theory.

\begin{figure}
	\centering
	\subfigure[Bootstrapping]
	{
		\includegraphics[width =2.2 in, height = 2.2in]{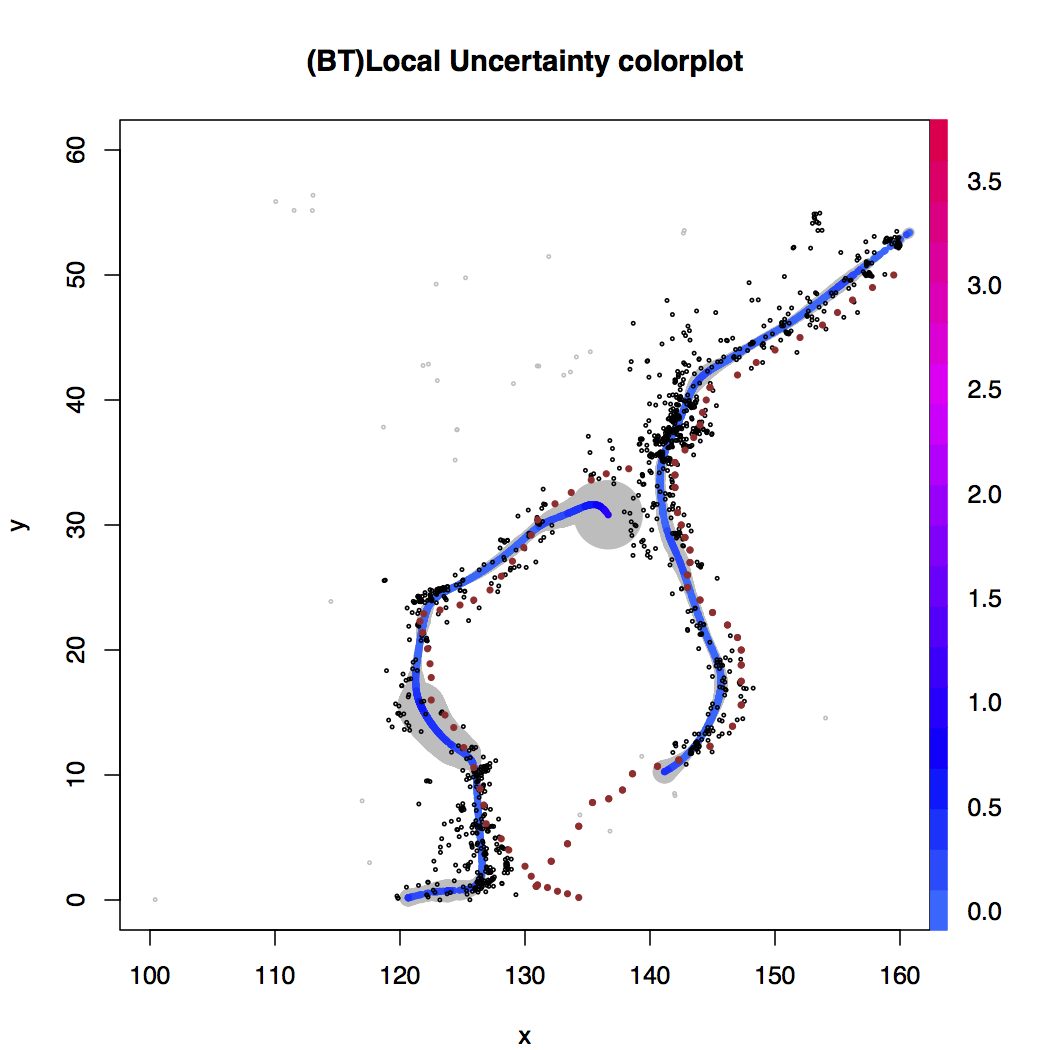}
	}
	\subfigure[Smooth bootstrapping]
	{
		\includegraphics[width =2.2 in, height = 2.2in]{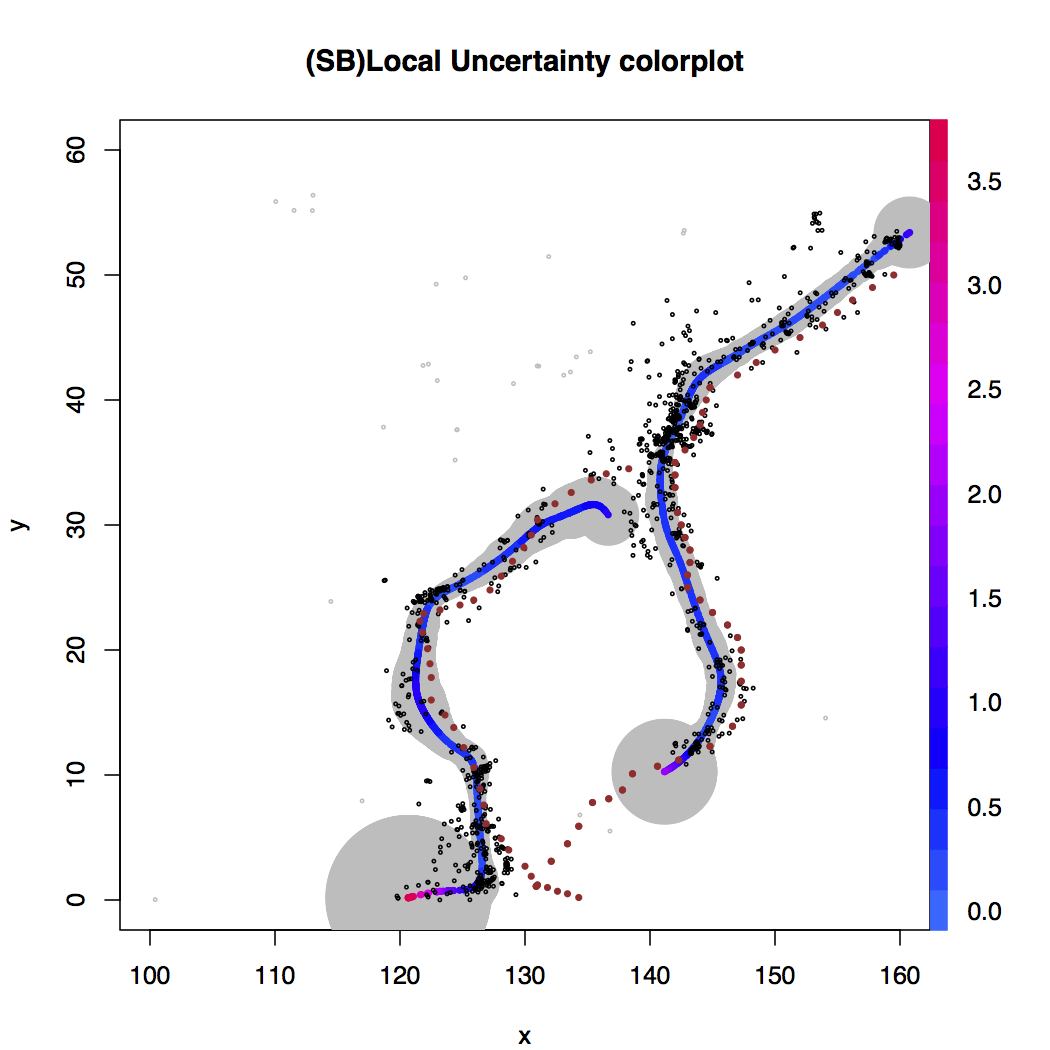}
	}
        \caption{Earthquake data. This is a collection of earthquake
          data in longitude $(100\sim160) E$, lattitude $(0\sim 60)N$
          from 01/01/2013 to 09/30/2013. Total sample size is
          $1169$. Blue curves are the estimated filaments; brown dots
          are the plate boundaries.}
\label{Fig:R2}
\end{figure}

\textit{Earthquake Data}.
We also apply our technique to data from the U.S.~Geological Survey
\footnote{The USGS dataset http://earthquake.usgs.gov/earthquakes/search/}
that locates $1,169$ earthquakes hat occur
in region between longitude $(100E\sim160E)$, latitude $(0N\sim 60N)$ 
and in dates between 01/01/2013 to 09/30/2013.
We are particularly interested in detecting
plate boundaries,
which see a high incidence of earthquakes.
We pre-process the data to remove a cluster of earthquakes that are
irrelevant to the plate boundary.
For this data, we only consider those filaments
with density larger than $\tau = 0.02$ of the maximum of the density.
Because the noise level is small, we adjust the KDE bandwidth
to $0.7$ times the Silverman rule ($h = 2.83$).

Figure~\ref{Fig:R2} displays the estimated filaments and $90\%$ pointwise
confidence sets.
The Figure shows the true plate boundaries from Nuvel data set
\footnote{Nuvel data set http://www.earthbyte.org/}
as brown points.
As can be seen in the Figure, smooth
bootstrapping has better coverage over the plate boundary. 
We notice the bad coverage in the bottom part; this is reasonable since the
boundary bias and lack of data cause trouble in estimation and
uncertainty measures.
We also identify some parts of filaments with high local uncertainty. 
The filaments with high uncertainty can be explained by theorem~\ref{LU2}. 
The data analysis again support our theoretical result.

In both Figure~\ref{Fig:R1},~\ref{Fig:R2}, we see a clear picture on the uncertainty assessment
for filament estimation. 
In data from two or three dimension, we can visualize uncertainties in estimation of
filaments with different colors or confidence regions.
That is, we can display estimation and the uncertainty in the same plot.

\section{Discussion and Future Work}

In this paper, we define a local uncertainty measure for filament estimation and
study its theoretical properties.
We apply bootstrap resampling 
to estimate local uncertainty measures and construct
confidence sets, and we prove that both are consistent and
data analysis also supports our result. 
Our method provides one way to numerically quantify
the uncertainty for estimating filaments. 
We also visualize uncertatiny measures with estimated filaments in the same plot;
this can be one easy way to show estimation and the uncertainty simultaneously.

Our approach has no constraints on
the dimension of the data so it can be extended to data from higher
dimension (although the confidence sets will be larger).
Our definition of local uncertainty and our estimation
method can be applied to other geometric estimation algorithms,
which we will investigate in the future.

\appendix

\section{Proofs}

\begin{proof}[ of Theorem \ref{S1}]
For the ridge set $R$, it is a
collection of solutions to $G(x) = V(x)V(x)^T \nabla p(x)$ as the
eigengap $\beta(x)>0$. But $V(x)$ is a $d\times (d-1)$ orthonormal
basis. So the solution to $G(x) = V(x)V(x)^T \nabla p(x)=0$ is equal
to the solution to $F(x) = V(x)^T \nabla p(x) =0$. Now $F(x): \R^d
\mapsto \R^{d-1}$. Hence, implicit function theorem tells us that
the differentiablilty of a local graph $\{(z,g(z)):z\in\R,
g(z)\in\R^{d-1}\}$ is the same as $F(x)$ when $\beta(x)>0$. Now
since the local graph is parametrized by one variable, we can
reparametrize it by a curve $\phi(x)$. And the differentiability of
the curve is the same as $F(x)$.

From a slight modification from theorem 3 in \cite{Genovese2012a}, the
$k$th order derivative of $F(x)$ depends on $k+2$th order derivative
of density if the eigengap $\beta(x)>0$.
Hence, if the density is $\mathbf{C}^k$ and we consider an open set
$U$ with $\beta(x)>0 \forall x \in U$, then we have $F(x)$ is
$\mathbf{C}^{k-2}$ on $U$ so the result follows.
\end{proof}

To prove theorem \ref{LU2}, we need the following lemmas:
\begin{lem}
\label{LUlem1}
Let $\hat{p}_n(x)$ be KDE for $p(x)$. Assume our kernel satisfies (K1), (K2). If $nh^{d+2}\rightarrow \infty, nh^{d+10}\rightarrow 0, h\rightarrow 0$. Then $\nabla \hat{p}_n(x)$ admits an asymptotic normal distribution by
\begin{align}
\sqrt{nh^{d+2}}(\nabla \hat{p}_n(x) -\nabla p(x) -B(x) h^2) \overset{d}{\rightarrow} N(0, \Sigma_0(x))
\end{align}
where
\begin{align}
B(x) &= \frac{m_2(K)}{2}\nabla (\nabla\bullet \nabla) p(x) \\
\Sigma_0(x) &=\nabla K (\nabla K)^T p(x) 
\end{align}
$K$ is the kernel used and $m_2(K)$ is a constant of kernel.
\end{lem}

\begin{proof}
For KDE $\hat{p}_n$,
\begin{align}
\hat{p}_n(x) = \frac{1}{n} \sum_{i=1}^n \frac{1}{h^d} K(\frac{x-X_i}{h}).
\end{align}
Hence for $\nabla \hat{p}_n$,
\begin{align*}
\nabla \hat{p}_n &= \frac{1}{n} \sum_{i=1}^n \frac{1}{h^d} \nabla K(\frac{x-X_i}{h})\\
& = \frac{1}{n} \sum_{i=1}^n \Phi(x;X_i).
\end{align*}
Notice that each $\Phi(x;X_i)$ is independent and identically distributed. 

We will show that $\Phi(x;X_i)$ satisfies conditions for Lyapounov's condition so that we have Central Limit Theorem (CLT) result for it. WLOG, we consider the third moment and focus on partial derivative over a direction, say $j$, we want
\begin{align}
\frac{(n\E(|\Phi_j(x;X_i)-\E(\Phi_j(x;X_i))|^3))^{\frac{1}{3}}}{(n Var(\Phi_j(x;X_i)))^{\frac{1}{2}}}\rightarrow 0
\end{align}
where $\Phi_j(x;X_i) = \frac{1}{h^d} \frac{\partial}{\partial x_j} K(\frac{x-X_i}{h})=\frac{1}{h^{d+1}} \frac{\partial K}{\partial x_j}(\frac{x-X_i}{h})$.

This is equivalent to show
\begin{align}
\frac{n^2\E(|\Phi_j(x;X_i)-\E(\Phi_j(x;X_i))|^3)^2}{n^3 Var(\Phi_j(x;X_i))^3} \rightarrow 0
\label{CLTeq1}
\end{align}

Now we put an upper bound on~\eqref{CLTeq1}, then we have
\begin{align*}
\frac{n^2\E(|\Phi_j(x;X_i)-\E(\Phi_j(x;X_i))|^3)^2}{n^3 Var(\Phi_j(x;X_i))^3} &\leq
\frac{n^2\E(|\Phi_j(x;X_i)|^3)^2}{n^3 Var(\Phi_j(x;X_i))^3}.
\end{align*}

We assume that $\int (\frac{\partial K}{\partial x_j}(u))^2du=C_2<\infty, \int (\frac{\partial K}{\partial x_j}(u))^3du=C_3<\infty $ for all $j=1,\cdots, d$. Therefore by Taylor expansion over density and take the first order, we have
\begin{align*}
\frac{n^2\E(|\Phi_j(x;X_i)|^3)^2}{n^3 Var(\Phi_j(x;X_i))^3} &=\frac{n^2 (\frac{C_3p(x)}{h^{2d+3}}+o(\frac{1}{h^{2d+3}}))^2}{n^3 (\frac{C_2 p(x)}{h^{d+2}}+o(\frac{1}{h^{d+2}}))^3}\\
&= O(\frac{1}{nh^d})\\
&= o(1)
\end{align*}
As a result, Lyapounov's condition is satisfied and this holds for all $j=1,\cdots, d$; so we have CLT for $\Phi(x;X_i)$.

By multivariate CLT we have
\begin{align*}
Var(\Phi(x;X_i))^{-\frac{1}{2}}[\nabla \hat{p}_n-\E(\Phi(x;X_i))] \overset{d}{\rightarrow} N(0, \mathbb{I}_d).
\end{align*}
where $\mathbb{I}_d$ is the identity matrix of dimension $d$.

By theorem 4 in~\cite{Chacon2011}, we have
\begin{align*}
\E(\Phi(x;X_i)) &= \nabla p(x)+ \frac{m_2(K)}{2}\nabla (\nabla\bullet \nabla) p(x) h^2 +O(h^4)\\
&= \nabla p(x)+ B(x)h^2 +O(h^4)\\
Var(\Phi(x;X_i)) &=\frac{1}{nh^{d+2}} \nabla K (\nabla K)^T p(x) +o(\frac{1}{nh^{d+2}})\\
&=\frac{\Sigma_0(x)}{nh^d} +o(\frac{1}{nh^d})
\end{align*}

Therefore, as $nh^{d+2}\rightarrow \infty$ and $h\rightarrow 0$, 
\begin{align*}
\sqrt{nh^{d+2}} \Sigma_0(x)^{-\frac{1}{2}}(\nabla \hat{p}_n(x) -\nabla p(x)- B(x)h^2 -O(h^4))\overset{d}{\rightarrow} N(0, \mathbb{I}_d)
\end{align*}

Now since $nh^{d+10}\rightarrow 0$ so $\sqrt{nh^{d+2}}O(h^4)$ tends to 0 and multiply $\Sigma_0(x)^{\frac{1}{2}}$ in both side, we get 
\begin{align*}
\sqrt{nh^{d+2}} (\nabla \hat{p}_n(x) -\nabla p(x)- B(x)h^2)\overset{d}{\rightarrow} N(0, \Sigma_0(x))
\end{align*}
This completes the proof.
\end{proof}

\begin{lem}
\emph{(\cite{Gine2002}; version of~\cite{Genovese2012a})}\\
\label{LUlem2}
Assume (K1), (P1) and the kernel function satisfies conditions in~\cite{Gine2002}. Then we have 
\begin{align}
||\hat{f}_{n,h}-f||_{\infty, k} = O(h^2) + O_P(\frac{\log n}{nh^{d+2k}}).
\end{align}
\end{lem}

\begin{lem}
\label{LUlem3}
For a density $p$, let $R$ be its filaments. For any points $x$ on $R$, let the Hessian at $x$ be $H(x)$ with eigenvectors $[v_1,\cdots, v_d]$ and eigenvalues $\lambda_1>0>\lambda_2\geq \cdots \lambda_d$. Consider any subspace $L$ spanned by a basis $[e_2,\cdots e_d]$ with $e_1$ be the normal vector for $L$. Then a sufficient and necessary condition for $x$ be a local mode of $p$ constrained in $L$ is
\begin{align}
\sum_{i=1}^d \lambda_i(v_i^Te_j)^2 <0, \forall j=2,\cdots, d.
\label{LUlem3eq1}
\end{align}
A sufficient condition for~\eqref{LUlem3eq1} is 
\begin{align}
(v_1^Te_1)^2 >\frac{\lambda_1}{\lambda_1-\lambda_2}.
\end{align}
\end{lem}

\begin{proof}
Let the Hessian of density $p$ at $x$ be $H(x)$ with eigenvectors $[v_1,\cdots, v_d]$ and associated eigenvalues $\lambda_1\geq\lambda_2\geq \cdots \lambda_d$. Consider any subspace $L$ spanned by a basis $[e_2,\cdots e_d]$ with $e_1$ be the normal vector for $L$

For any $x$ on the ridge, we have $\lambda_1>0>\lambda_2$. $x$ is the mode constrained in the subspace $L$ if $\nabla_L \nabla_L p(x)$ is negative definite. By spectral decomposition, we can write
\begin{align*}
\nabla_L \nabla_L p(x) &= L^T H(x) L\\
& = L^TU(x)\Omega(x)U(x) L^T,
\end{align*}
where $U(x) = [v_1,\cdots, v_d]$ and $\Omega(x)$ is a diagonal matrix of eigenvalues.
So this matrix will be negative definite if and only if all its diagonoal elements are negative.

That is, the sufficient and necessary condition is
\begin{align}
(\nabla_L \nabla_L p(x))_{ii}<0, i=1,\cdots, d-1.
\label{LUlem3eq2}
\end{align}

We explicitly derive the form of~\eqref{LUlem3eq2} and consider the sufficient and necessary condition:
\begin{align*}
(\nabla_L \nabla_L p(x))_{jj} & = ( L^TU(x)\Omega(x)U(x) L^T)_{jj}\\
&= \sum_{i=1}^d e_j^T v_i \lambda_i v_i^T e_j \\
&= \sum_{i=1}^d \lambda_i (v_i^T e_j)^2<0, j=2,\cdots d.
\end{align*}
So we prove the first condition.

To see the sufficient condition, we note that by definition, $\lambda_2(x)\geq\cdots\geq\lambda_d(x)$. So for each $j$,
\begin{align*}
\sum_{i=2}^d \lambda_i (v_i^T e_j)^2\leq \lambda_2 \sum_{i=2}^d  (v_i^T e_j)^2.
\end{align*}
This implies that for each $j$,
\begin{align*}
\sum_{i=1}^d \lambda_i (v_i^T e_j)^2 &\leq (\lambda_1-\lambda_2)  (v_1^T e_j)^2 +\lambda_2 \sum_{i=1}^d  (v_i^T e_j)^2\\
&= (\lambda_1-\lambda_2)  (v_1^T e_j)^2 +\lambda_2<0.
\end{align*}
Note that we use the fact $\sum_{i=1}^d  (v_i^T e_j)^2 = 1$ since $|e_j|=1$ and $v_i$'s are basis. 

Since both $e_1,\cdots,e_d$ and $v_1,\cdots,c_d$ are basis, we have
\begin{align*}
(v_1^T e_j)^2 \leq 1-(v_1^T e_1)^2
\end{align*}
for all $j=2,\cdots,d$. 
Then we further have
\begin{align*}
(\lambda_1-\lambda_2)  (v_1^T e_j)^2 +\lambda_2 &\leq  (\lambda_1-\lambda_2)  (1-(v_1^T e_1)^2) +\lambda_2\\
& = -(\lambda_1-\lambda_2)(v_1^T e_1)^2 +\lambda_1
\end{align*}
for each $j=2,\cdots d$. 

Putting altogether, we have
\begin{align*}
(\nabla_L \nabla_L p(x))_{jj} &= \sum_{i=2}^d \lambda_i (v_i^T e_j)^2\\
&\leq(\lambda_1-\lambda_2)  (v_1^T e_j)^2 +\lambda_2\\
&\leq  -(\lambda_1-\lambda_2)(v_1^T e_1)^2 +\lambda_1<0
\end{align*}
for all $j=2,\cdots, d$

So a sufficient condition is 
\begin{align*}
(v_1^T e_1)^2> \frac{\lambda_1}{(\lambda_1-\lambda_2)}.
\end{align*}
\end{proof}

\begin{proof}[ of Theorem \ref{LU2}]
By definition of $\phi(s), \hat{\phi}(s)$ and the nature of ridge, $\phi(s)$ is the local mode of $p(x)$ in the subspace spanned by $L(s)$ near $\phi(s)$. 
By lemma~\ref{LUlem3}, $\hat{\phi}(s)$ will also the local mode of $\hat{p}_n(x)$ in the subspace spanned by $L(s)$ near $\phi(s)$ once the two filamnets are closed enough and the local direction of filaments are also closed. This will be shown in 2. of theorem~\ref{SB2}.

Hence, we have
\begin{align*}
\nabla_{L(s)} p(\phi(s))= \nabla_{L(s)} \hat{p}_n(\hat{\phi}_n(s)) =0.
\end{align*}

Now applying Taylor expansion for $\nabla_L(s) \hat{p}_n(\hat{\phi}_n(s))$ near $\phi(s)$, we have
\begin{align*}
0 &= \nabla_{L(s)} \hat{p}_n(\hat{\phi}_n(s))\\
&= \nabla_{L(s)} \hat{p}_n(\phi(s)) + \hat{H}^{L(s)}_n(\phi^*(s)) L(s)^T (\hat{\phi}_n(s)-\phi(s))
\end{align*}
where $\hat{H}_n(x;L(s)) $ is the projected Hessian of KDE while $\phi^*(s) = t\phi(s)+(1-t)\hat{\phi}(s)$, for some $0\leq t\leq 1$.

Accordingly, 
\begin{align}
L(s)^T (\hat{\phi}_n(s)-\phi(s) )= -\hat{H}_n(\phi^*(s);L(s))^{-1} \nabla_{L(s)} \hat{p}_n(\phi(s)).
\label{eq2}
\end{align}

By lemma \ref{LUlem2} with $k=2$ and we pick a suitable $h=h_n$, we have
\begin{align*}
\hat{H}_n(\phi^*(s)) \overset{p}{\rightarrow} H(\phi^*(s))
\end{align*}
which implies
\begin{align*}
\hat{H}_n(\phi^*(s);L(s)) \overset{p}{\rightarrow} H(\phi^*(s);L(s))
\end{align*}
and $\phi^*(s)$ will converge to $\phi(s)$. Consequently,
\begin{align*}
\hat{H}_n(\phi^*(s);L(s)) \overset{p}{\rightarrow} H(\phi(s);L(s)).
\end{align*}
This implies 
\begin{align}
\hat{H}_n(\phi^*(s);L(s))^{-1} \overset{p}{\rightarrow} H(\phi(s);L(s))^{-1}
\label{UTeq1}
\end{align}
since $H(x)$ is non-singular.

Now we consider $\nabla_{L(s)} \hat{p}_n(\phi(s))$. Recall that we assume $nh^{d+2}\rightarrow \infty, nh^{d+10}\rightarrow 0, h\rightarrow 0$; by lemma \ref{LUlem1} we have
\begin{align*}
\sqrt{nh^{d+2}}( \nabla\hat{p}_n(\phi(s))-\underbrace{\nabla p(\phi(s))}_\text{ =0  (local mode)} - B(\phi(s))h^2)\notag\\
 \overset{d}{\rightarrow} N(0, \Sigma_0(\phi(s)))
\end{align*}
with
\begin{align*}
B(\phi(s)) &= \frac{m_2(K)}{2}\nabla (\nabla\bullet \nabla) p(\phi(s))\\
\Sigma_0(s) &=\nabla K (\nabla K)^T p(\phi(s)) .
\end{align*}

Hence, for the subspace case:
\begin{align}
\sqrt{nh^{d+2}}(L(s)^T\nabla \hat{p}_n(\phi(s)) - L(s)^TB(\phi(s))h^2)\notag \\
\overset{d}{\rightarrow} N(0, L(s)\Sigma_0(\phi(s))L(s)^T)
\label{eq3}
\end{align}

Recalled \eqref{eq2} :
\begin{align*}
L(s)^T (\hat{\phi}_n(s)-\phi(s) ) &=-\hat{H}_n(\phi^*(s);L(s))^{-1} \nabla_{L(s)} \hat{p}_n(\phi(s))\\
&=\hat{H}_n(\phi^*(s);L(s))^{-1} L(s)^T\nabla \hat{p}_n(\phi(s))
\end{align*}
now we plug in~\eqref{eq3} and apply Slutsky's theorem, we get
\begin{align}
&\sqrt{nh^{d+2}}[ L(s)^T (\hat{\phi}_n(s)-\phi(s) ) - \mu(s) h^2]
\overset{d}{\rightarrow} N(0, \Sigma(s))
\label{eq4}
\end{align}
where
\begin{align*}
\mu(s) &= H(\phi(s);L(s))^{-1}L(s)^TB(\phi(s)) \\
\Sigma(s) &= H(\phi(s);L(s))^{-1}L(s)\Sigma_0(\phi(s))L(s)^TH(\phi(s);L(s))^{-1}.
\end{align*}

Now since $\hat{\phi}_n(s)-\phi(s)$ always lays in the subspace $L(s)$, we have
\begin{align*}
L(s) L(s)^T (\hat{\phi}_n(s)-\phi(s) ) = \hat{\phi}_n(s)-\phi(s).
\end{align*}

Consequent, we can multiply $L(s)$ in \eqref{eq4} to obtain
\begin{align*}
\sqrt{nh^{d+2}}[  (\hat{\phi}_n(s)-\phi(s) ) - L(s)\mu(s) h^2] \overset{d}{\rightarrow} L(s) A(s)
\end{align*}
with 
\begin{align*}
A(s) \overset{d}{=}N(0, \Sigma(s))\in\R^{d-1}
\end{align*}
is a Gaussian process in $\R^{d-1}$.
\end{proof}

\begin{proof}[ of Theorem \ref{SB2}]
1. By assumption, the ridges for $p$ and $q_m$ have positive conditioning number.  We can apply theorem 6. in Genovese et. al. (2012) so that $R(p)$ and $R(q_m)$ will be asymptotically topological homotopy. So we can always find a continuous bijective mapping to map every point on $R(p)$ to $R(q_m)$. We define $\xi_m$ be such a map on each point of $R(p)$. Since the Hausdorff distance converge to $0$, the associated mapping can be picked such that each pair $\phi(s), \psi_m(\xi_m(s)$ has distance less than Hausdorff distance. So the result follows.

2. Recall that filaments are solutions to 
\begin{align*}
\{x: V(x)V(x)^T \nabla p(x)=0, \beta(x)>0\}.
\end{align*}
The direction of ridge ($\phi(s), \psi'_m(s)$) depends on up to third derivative at points on the filaments. From 1., we know that the location of $\psi_m(\xi_m(\phi(s)))$ will converge to $\phi(s)$ and we have the uniform convergence up to the third derivative by assumptions. Hence, by uniformly convergence and both $p,q_m$ are $\mathbf{C}^{d+3}$, we have convergence in the tangent line at each point of filaments. So this implies the inner product to be 1.

3. From theorem~\ref{LU2}, $\mu(s) = c(K) H(\phi(s);L(s))^{-1} \nabla_{L(s)} (\nabla \bullet \nabla) p(\phi(s))$. By assumption, we have uniform convergence up to third derivative and by 2. we have convergence in subspace. So $\mu(s)$ from $q_m$ will unifromly converge to that from $p$.

4. Similar to 3.
\end{proof}

\newpage




\end{document}